\documentclass[11pt,fullpage]{article}

\usepackage{graphicx}
\usepackage{pgf}
\usepackage{latexsym}
\usepackage{color}
\usepackage{multirow}

\usepackage{amsmath}
\usepackage{amssymb}
\usepackage{amsthm}
\usepackage{amsfonts}
\usepackage{bm}
\usepackage{epstopdf}

\renewcommand{\baselinestretch} {1.3}

\makeatletter 
\def\singlespace{\def\baselinestretch{1}\@normalsize}

\newtheorem{theorem}{Theorem}
\newtheorem{lemma}{Lemma}

\newtheorem{corollary}{Corollary}

\newcommand{\cM}{{\cal M}}

\newcommand{\cH}{\cal H}
\newcommand{\cG}{\cal G}
\newcommand{\Roracle}{R({\rm oracle})}
\newcommand{\te}{\theta}
\newcommand{\eps}{\varepsilon}
\def\*#1{\mathbf{#1}}

\newcommand{\hatS}{\widehat{S}}
\newcommand{\hatte}{\hat{\te}}
\newcommand{\hf}{\hat{f}}
\newcommand{\hte}{\hat{\te}}
\newcommand{\hatf}{\hat{f}}

\newcommand{\Tr}{\mbox{Tr}}
\newcommand{\Pen}{\mbox{Pen}}

\newcommand{\Var}{\mbox{Var}}
\newcommand{\card}{\mbox{card}}

\long\def\ignore#1{}

\newcommand{\be}{\begin{equation}}
\newcommand{\ee}{\end{equation}}
\newcommand{\beqn}{\begin{eqnarray}}
\newcommand{\eeqn}{\end{eqnarray}}
\newcommand{\bes}{\begin{equation*}}
\newcommand{\ees}{\end{equation*}}
\newcommand{\beqns}{\begin{eqnarray*}}
\newcommand{\eeqns}{\end{eqnarray*}}

\newcommand{\lkr}{\left(} 
\newcommand{\lkv}{\left[} 
\newcommand{\rkv}{\right]} 
\newcommand{\rkr}{\right)}  
\newcommand{\lfi}{\left\{}  
\newcommand{\rfi}{\right\}} 

\graphicspath{%
    {converted_graphics/}
    {/}
}
\begin{document}

\title{\bf Solution of linear ill-posed problems by model selection and aggregation}

\author{{\bf Felix Abramovich}\\
Department of Statistics\\
 and Operations Research \\
Tel Aviv University \\
Israel \\
{\it felix@post.tau.ac.il}
\and
{\bf Daniela De Canditiis}\\	
Istituto per le Applicazioni  \\
del Calcolo "M. Picone", CNR  \\
Rome\\
 Italy \\
{\it d.decanditiis@iac.cnr.it}
\and {\bf Marianna Pensky} \\
Department of Mathematics\\
University of Central Florida \\
USA \\
{\it Marianna.Pensky@ucf.edu}
}

\date{}

\maketitle

\begin{abstract}
We consider  a general statistical linear inverse problem,
where the solution is represented via a known (possibly overcomplete) dictionary
that allows its sparse representation. We propose two different approaches.
A model selection estimator selects a single model by
minimizing the penalized empirical risk over all possible models.
By   contrast with direct problems, the penalty depends on the model 
itself rather than on its size only as for complexity penalties.
A Q-aggregate estimator averages over the entire collection of
estimators with properly chosen weights. Under mild
conditions on the dictionary, we establish oracle inequalities both 
with high probability and in expectation for the two estimators.
Moreover, for the latter estimator these inequalities are sharp.
The proposed procedures are implemented numerically and
their performance is assessed by a simulation study.
\end{abstract}

\vspace{2mm}
{\em AMS (2000) Subject Classification}: {Primary: 62G05.  Secondary:  62C10  }

\bigskip


\section{Introduction} \label{sec:intr}

Linear inverse problems, where the data is available not on the object of primary
interest but only in the form of its   linear transform, appear in a variety of fields:
medical imaging (X-ray tomography, CT and MRI), astronomy (blured images), finance (model
calibration of volatility) to mention just a few. The main difficulty in solving
inverse problems is due to  the fact that most of practically interesting and
relevant cases fall into into the category of so-called ill-posed problems, where
the solution cannot be obtained numerically by simple invertion of the transform.
In statistical inverse problems the data is, in addition, corrupted by random noise that makes the solution even more challenging.

Statistical linear inverse problems have been intensively studied and there exists an enormous amount of literature devoted to
various theoretical and applied aspects of their solutions. We  refer a reader to
to Cavalier (2009) for review and references therein.

Let $\cG$ and $\cH$ be two separable Hilbert spaces and $A: \cG \rightarrow \cH$ be a bounded linear operator.
Consider a general statistical linear inverse problem
\be \label{eq:model}
y = Af + \sigma \varepsilon,
\ee
where $y$ is the observation, $f  \in G$ is the (unknown) object of interest,
$\varepsilon$ is a white noise and $\sigma$ is a (known) noise level. For ill-posed
problems $A^{-1}$ does not exist as a linear bounded operator.

Most of approaches for solving (\ref{eq:model}) essentially rely on reduction of the original
problem to a sequence model using the following general scheme:
\begin{enumerate}
\item Choose  some orthonormal
basis $\{\phi_j\}$ on $\cG$ and expand the unknown $f$ in (\ref{eq:model}) as
\be \label{eq:series}
f=\sum_j \langle  f,\phi_j\rangle_G~ \phi_j
\ee

\item Define $\psi_j$ as the solution of $A^*\psi_j=\phi_j$,
where $A^*$ is the adjoint operator, that is, $\psi_j=A(A^*A)^{-1}\phi_j$. Reduce
(\ref{eq:model}) to the equivalent sequence model:
\be \label{eq:sequence}
\langle  y,\psi_j\rangle_H~=~\langle  Af,\psi_j\rangle_H+\langle \varepsilon,\psi_j\rangle_H~
=~\langle f,\phi_j\rangle_G+\langle \varepsilon,\psi_j\rangle_H,
\ee
where for ill-posed problems $\Var\left(\langle y,\psi_j\rangle_H\right)= \sigma^2 \|\psi_j\|_H^2$ increases with $j$.
Following the common terminology, an inverse problem is called {\em mildly ill-posed},
if the variances increase polynomially and {\em severely ill-posed} if their growth
is exponential (see, e.g., Cavalier, 2009).

\item Estimate the unknown coefficients $\langle f,\phi_j\rangle_G$ from empirical coefficients $\langle y,\psi_j\rangle_H$:
$\widehat{\langle f,\phi_j\rangle_G}=\delta\{\langle y,\psi_j\rangle_H\}$, where $\delta(\cdot)$ is some truncating/shrinking/thresholding
procedure (see, e.g., Cavalier 2009, Section 1.2.2.2 for
a survey), and reconstruct $f$ as
$$
\hat{f}=\sum_j\widehat{\langle f,\phi_j\rangle_G}~ \phi_j
$$
\end{enumerate}

Efficient representation of $f$ in a chosen basis $\{\phi_j\}$ in (\ref{eq:series}) is essential.
In the widely-used singular value decomposition (SVD), $\phi_j$'s
are the orthogonal eigenfunctions of the self-adjoint operator $A^*A$ and
$\psi_j = \lambda_j^{-1} A\phi_j$, where $\lambda_j$ is the corresponding eigenvalue.
SVD estimators are known to be optimal in various minimax settings
over certain classes of functions (e.g., Johnstone \& Silverman, 1990; Cavalier
\& Tsybakov, 2002;  Cavalier {\em et al.}, 2002).  A serious drawback of SVD is that the basis is
defined entirely by the operator $A$ and ignores the specific properties of
the object of interest $f \in \cG$. Thus, for a given $A$, the same basis will be
used regardless of the nature of a scientific problem at hand. While the SVD-basis
could be very efficient for representing $f$ in one area, it might yield poor
approximation in the other. The use of SVD, therefore, restricts one within certain
classes $\cG$ depending on a specific operator $A$. See Donoho (1995) for further discussion.

In wavelet-vaguelette decomposition (WVD), $\phi_j$'s are orthonormal wavelet
series. Unlike SVD-basis, wavelets allows sparse representation for various classes of functions and
the resulting WVD estimators have been studied in Donoho (1995), Abramovich \&
Silverman  (1998), Kalifa \& Mallat (2003), Johnstone {\em et al.} (2004). 
However, WVD imposes relatively stringent conditions on $A$
that are satisfied only for specific types of operators, mainly of convolution type. 

A general  shortcoming of orthonormal  bases  is due to the fact that they
may be ``too coarse'' for efficient representation of unknown $f$. Since 90s, there was a growing interest in the
atomic decomposition of functions over {\em overcomplete} dictionaries
(see, for example, Mallat \& Zhang, 1993;
Chen, Donoho \& Saunders, 1999; Donoho \& Elad, 2003).
Every basis is essentially only a {\em minimal} necessary dictionary. Such ``miserly''
representation usually causes poor adaptivity (Mallat \& Zhang, 1993).
Application of overcomplete dictionaries improves adaptivity of the representation, because one can
choose now the most efficient (sparse) one among many available. One can see here an interesting
analogy with colors. Theoretically, every other color can be generated by
combining three basic colors (green, red and blue) in corresponding proportions.
However, a painter would definitely prefer to use the whole available palette
(overcomplete dictionary) to get the hues he needs!
Selection of appropriate subset of atoms (model selection) that allows a sparse
representation of a solution is a core element in such an approach. Pensky (2016)
was probably the first to use overcomplete dictionaries for solving inverse problems.
She utilized the Lasso techniques for model
selection within the overcomplete dictionary, established oracle inequalities with high probability and
applied the proposed procedure to several examples of linear inverse problems.
However, as usual with Lasso, it required
somewhat restrictive compatibility conditions on the design matrix $\Phi$.

In this paper we propose two alternative approaches for overcomplete dictionaries based estimation in  linear ill-posed problems.  
The first estimator is obtained  by minimizing penalized empirical risk with a penalty on model $M$ proportional 
to $\sum_{j \in M} \|\psi_j^2\|$.  The second one is based on a Q-aggregation type procedure that is specifically 
designed for solution of linear ill-posed problems and is different  from that of Dai {\em et al.} (2012) 
developed for solution of direct problems. We establish   oracle inequalities for both estimators that  
hold  with high probabilities and in expectation. Moreover, for the Q-aggregation estimator, the 
inequalities are sharp. 
Simulation study shows that the new techniques produce more accurate estimators than Lasso.

The rest of the paper is organized as follows. Section \ref{sec:notations} introduces the
notations and some preliminary results. The model selection and
aggregation-type procedures are studied respectively in Section \ref{sec:model_select} 
and Section \ref{sec:aggregation}. The  simulation study is described in Section \ref{sec:simulations}. All   
proofs are given in the Appendix.


%
\section{Setup and notations} \label{sec:notations}

Consider a discrete analog of a general statistical linear inverse problem (\ref{eq:model}):
\be \label{eq:discrete}
y=Af+\sigma \varepsilon,
\ee
where $y \in \mathbb{R}^n$ is the vector of observations, $f \in \mathbb{R}^n$ is
the unknown vector to be recovered, $A$ is a known (ill-posed) $n \times m\;(n \geq m)$
matrix with $rank(A)=m$, and $\varepsilon \sim N(0,I_n)$.

In what follows $\|\cdot\|$ and $\langle  \cdot,\cdot \rangle$ denote respectively a
Euclidean norm and an inner product in $\mathbb{R}^n$. Let $\phi_j \in \mathbb{R}^n,\;j=1,\cdots,p$
with $\|\phi_j\|=1$ be a set of normalized vectors (dictionary),
where typically $p > n$ (overcomplete dictionary). Let $\Phi_{n \times p}$ be the
complete dictionary matrix with the columns $\phi_j,\;j=1,\ldots,p$, and
$\Psi_{n \times p}$ is such that $A^T \Psi = \Phi$, that is, $\Psi=A(A^T A)^{-1} \Phi$ and
$\psi_j=A(A^T A)^{-1}\phi_j$.

In what follows, we assume that, for some positive $r$, the  minimal   $r$-sparse eigenvalue of $\Phi^T \Phi$  is
separated from  zero:
$$
\nu^2_r=\min_{x \in \mathbb{R}^n, \|x\|_0 \leq r} \frac{\|\Phi x\|^2} {\|x\|^2} >0
$$
and consider a set of models $\cM=\{M \subseteq \{1,\ldots,p\}: |M| \leq r/2\}$ of sizes not larger
than $r/2$.

For a given model $M \subseteq \{1,\ldots,p\}$ define a $p \times p$ diagonal indicator matrix
$D_M$ with diagonal entries  $d_{Mj}=I\{j \in M\}$.
The design matrix corresponding to $M$ is then $\Phi_M=\Phi D_M$, while
$\Psi_M=\Psi D_M=A(A^T A)^{-1} \Phi_M$.
Let $H_M=\Phi_M(\Phi_M^T \Phi_M)^{-1}\Phi_M^T$ be the projection matrix on a span of nonzero columns of $\Phi_M$ and
$$
f_M=H_M f =\Phi_M(\Phi_M^T \Phi_M)^{-1}\Phi_M^T~ f = \Phi_M(\Phi_M^T \Phi_M)^{-1}\Psi_M^T~ Af
$$
be the projection of $f$ on $M$. Denote
\be \label{eq:z_def}
z = (A^T A)^{-1}A^T y = f + \xi,\; {\rm where} \; \xi = (A^T A)^{-1}A^T \eps.
\ee

Consider the corresponding projection estimator
\be \label{eq:hatfM}
\hat{f}_M=H_M z= H_M (A^T A)^{-1}A^T y=\Phi_M(\Phi_M^T \Phi_M)^{-1}\Psi_M^T~y=\Phi_M \hat{\theta}_M,
\ee
where the vector of projection coefficients $\hat{\theta}_M=(\Phi_M^T \Phi_M)^{-1}\Psi_M^T~y$.
By straightforward calculus, $\hat{f}_M \sim N(f_M,\sigma^2H_M (A^T A)^{-1}H_M)$ and
the quadratic risk
\be \label{eq:quadratic_risk}
E\|\hat{f}_M-f\|^2=\|f_M-f\|^2+\sigma^2 \Tr\left((A^T A)^{-1}H_M\right)
\ee
The oracle model is the one that minimizes (\ref{eq:quadratic_risk}) over all
models $M \in \cM$ and the ideal oracle risk is
\be \label{eq:oracle_risk}
\Roracle=\inf_{M \in \cM} E\|\hat{f}_M-f\|^2=\inf_{M} \left\{\|f_M-f\|^2+\sigma^2 \Tr\left((A^T A)^{-1}H_M\right)\right\}
\ee
The oracle risk is obviously unachievable but can be used as
a benchmark for a quadratic risk of any available estimator.


\section{Model selection by penalized empirical risk}
\label{sec:model_select}

\ignore{
We consider only models $M \in \cM=\{M: |M| \leq r\}$ of sizes not larger
than $r$ since otherwise, there necessarily exists another model $M'$ of size at most
$r$ such that $f_M=f_{M'}$.
}
For a given model $M \in \cM$, $\hat{f}_M$ in (\ref{eq:hatfM}) minimizes the corresponding empirical risk
$\|z-\tilde{f}_M\|^2$, where $z$ was defined in (\ref{eq:z_def}). Select
a model $\widehat{M}$ by minimizing the penalized empirical risk:
\be \label{eq:mhat}
\widehat{M}=\underset{M \in \cM}{\operatorname{argmin}}\, \left\{ \|z-\hat{f}_M\|^2 +\Pen(M)\right\}
= \underset{M \in \cM}{\operatorname{argmin}}\, \left\{-\|\hat{f}_M\|^2 + \Pen(M) \right\},
\ee
where $\Pen(M)$ is a penalty function on a model $M$.
The proper choice of $\Pen(M)$ is the core of such an approach.

For {\em direct} problems ($A=I$), the penalized empirical risk approach, 
with the complexity type   penalties $\Pen(|M|)$ on a model size, has been
intensively studied in the literature. In the last decade, in the context of linear regression,
the  in-depth theories (risk bounds, oracle inequalities, minimaxity) 
have been developed    by a number of authors (see, e.g., 
Foster \& George (1994), Birg\'e \& Massart (2001, 2007), Abramovich \&
Grinshtein (2010), Rigollet \& Tsybakov (2011), Verzelen (2012) among many others).

For inverse problems, Cavalier {\em et al.} (2002) considered
a truncated orthonormal series estimator, where the cut-off point was chosen
by SURE criterion corresponding to the AIC-type
penalty $\Pen(M)=2\sigma^2 \Tr((A^T A)^{-1}H_M)$ and established oracle inequalities for the resulting estimator $\hat{f}_{\widehat M}$.
It was further generalized and improved by risk hull minimization in Cavalier \& Golubev (2006).

To the best of our knowledge, Pensky (2016) was the first to consider model selections within
{\em overcomplete dictionaries} by empirical risk minimization for statistical inverse problems.
She utilized Lasso penalty. However, as usual with Lasso, it required
somewhat restrictive compatibility conditions on the design matrix $\Phi$ (see
Pensky, 2016 for more details).

In this paper, we utilize the penalty $\Pen(M)$ that depends on the Frobenius norm
of the matrix $\Psi_M$:
$$
\|\Psi_M\|^2_F=\sum_{j \in M} \|\psi_j\|^2=\Tr\left((A^T A)^{-1}\Phi_M\Phi_M^T\right)
$$
The following theorem provides nonasymptotic upper bounds for the quadratic risk
of the resulting estimator $\hat{f}_{\widehat M}$ both with high probability and in expectation:

\begin{theorem} \label{th:oracle_mod_select}
Consider the model (\ref{eq:discrete}) and the penalized empirical risk estimator
$\hat{f}_{\widehat M}$, where the model ${\widehat M}$ is selected w.r.t. (\ref{eq:mhat}) with the penalty
\be \label{eq:penalty}
\Pen(M) \geq \frac{4 \sigma^2 (\delta + 1)}{a \nu^2_r}\ \|\Psi_M\|^2_F\ \ln p
\ee
for some $\delta > 0$ and $0 < a < 1$. Then,

\begin{enumerate}

\item
With probability at least
$1 - \sqrt{\frac{2}{\pi}}~p^{-\delta}$
\be \label{eq:oraclehp}
\|\hat{f}_{\widehat M}-f\|^2 \leq \frac{1+a}{1-a}~ \min_{M \in \cM}
\left\{ \|f_M - f\|^2 + \frac{a^2+a+4}{2(1+a)}~ \Pen(M) \right\}
\ee

\item If, in addition, we restrict the set of admissible models to
$\cM_\gamma=\{M \in \cM:\|\Psi_M\|_F^2 \leq \gamma^2 n\}$ for some constant $\gamma$,
\be \label{eq:oraclee}
E\|\hat{f}_{\widehat M}-f\|^2  \leq \frac{1+a}{1-a}~ \min_{M \in \cM_\gamma}
\left\{ \|f_M - f\|^2 +
\frac{3}{2(1+a)}~ \Pen(M) \right\} + \frac{4 \gamma^2 \sigma^2}{a(1-a)\nu^2_r}~ n p^{-\delta}
\ee

\end{enumerate}

\end{theorem}

The additional restriction on the set of models $\cM$ in the second part of Theorem \ref{th:oracle_mod_select}
is required to guarantee that the oracle risk in (\ref{eq:oracle_risk}) does not grow faster than $n$.

Note that for the direct problems, $\Psi_M=\Phi_M$ and the penalty (\ref{eq:penalty})
is the RIC-type complexity penalty of Foster \& George (1994) of the form
$\Pen(M)=C |M| \ln p$.

We can compare the quadratic risk of the proposed estimator with the oracle
risk $\Roracle$ in (\ref{eq:oracle_risk}). Consider the penalty
\be \label{pen_exact}
\Pen(M)=\frac{4 \sigma^2 (\delta + 1)}{a \nu^2_r}\ \|\Psi_M\|^2_F\ \ln p
\ee
for some $0 < a < 1$ and $\delta \geq 2$.
Assume that $p\geq n$ (overcomplete dictionary) and choose $\delta \geq 2$. Then,  the last term in the RHS of
(\ref{eq:oraclee}) turns to be of a smaller order and we obtain
$$
E\|\hat{f}_{\widehat M}-f\|^2 \leq C_1\min_{M \in \cM_\gamma}\left\{ \|f_M - f\|^2 +
C_2 \frac{\gamma \sigma^2}{\nu_r^2} \|\Psi_M\|^2_F \ln p\right\}
$$
for some positive constants $C_1, C_2$ depending on $a$ and $\delta$ only. By standard linear
algebra arguments, $\|\Psi_M\|^2_F=\Tr\left((A^T A)^{-1}\Phi_M \Phi_M^T\right)
\leq \kappa^2_r~ \Tr\left((A^T A)^{-1}H_M\right)$ and, therefore, the following
oracle inequality holds:

\begin{corollary} \label{cor:cor1}
Assume that $p \geq n$ and consider the penalized empirical risk estimator
${\widehat M}$ from Theorem \ref{th:oracle_mod_select}, where ${\widehat M}$ is selected
w.r.t. (\ref{eq:mhat}) over $\cM_\gamma$  with the penalty \eqref{pen_exact}
for some $0 < a < 1$ and $\delta \geq 2$. Then,
$$
E\|\hat{f}_{\widehat M}-f\|^2 \leq C_0 \frac{\kappa^2_r}{\nu^2_r} \ln p~ \Roracle
$$
for some constant $C_0>0$ depending on $a$, $\delta$ and $\gamma$ only.
\end{corollary}

\noindent
Thus, the quadratic risk of the proposed estimator $\hat{f}_{\widehat M}$ is
within $\ln p$-factor of the ideal oracle risk. 
 The $\ln p$-factor is a common closest rate at which an estimator can approach an oracle   
even in direct   (complete) model selection problems 
%
(see, e.g., Donoho \& Johnstone, 1994;
Birg\'e \& Massart, 2001; Abramovich \& Grinshtein, 2010; Rigollet \& Tsybakov, 2011;
and also Pensky (2016) for inverse problems). For an {\em ordered} model selection
within a set of nested models, it is possible to construct estimators that achieve the oracle risk within a constant factor
(see, e.g., Cavalier {\em et al.},  2002 and Cavalier \& Golubev, 2006).

Similar oracle inequalities (even sharp with the coefficient in front
of $\|f_M - f\|^2$ equals to one) with high probability were obtained for the Lasso estimator but under the additional compatibility assumption on the
matrix $\Phi$ (Pensky, 2016).


\section{Q-aggregation}
\label{sec:aggregation}

Note that inequalities \eqref{eq:oraclehp} and \eqref{eq:oraclee} in Theorem \ref{th:oracle_mod_select} for model selection estimator
are not sharp in the sense that the coefficient in front of the minimum is greater than one.
In order to  derive sharp oracle inequalities both in probability and expectation, one needs to aggregate the entire collection of
estimators: $\hat{f}=\sum_{M \in \cM} \te_M \hat{f}_M$ rather than to select a single estimator $\hat{f}_{\widehat M}$.

Leung \& Barron (2006) considered exponentially weighted aggregation (EWA) with
$\theta_M \propto \pi_M \exp\{\|z-\hat{f}_M\|^2/\beta\}$, where $\pi$ is some  (prior)
probability measure on $\cM$ and $\beta>0$ is a tuning parameter.
They established sharp oracle inequalities in expectation for EWA in direct problems.
Dalalyan \& Salmon (2012) proved sharp oracle inequalities in expectation for EWA of affine
estimators in nonparametric regression model. Their paper offers limited extension of the theory
to the case of  mildly ill-posed inverse problems where $\Var(\langle y,\psi \rangle_H)=\sigma^2 \|\psi_j\|^2_H$
increase at most polynomially with $j$. However, their results are valid only for the SVD decomposition
and require block design which seriously limits the scope of application of their theory.
Moreover, Dai {\em et al.} (2012) argued that EWA cannot satisfy sharp oracle inequalities with high probability and
proposed instead to use Q-aggregation.

Define a general Q-aggregation estimator of $f$ as
\be \label{eq:fhattheta}
 \hat{f}_{\hat \te}=\sum_{M \in \cM}  \te_M \hat{f}_M,
\ee
where the vector of weights $\hat{\te}$ is the solution of the following optimization
problem:
\be \label{eq:Q_weights}
\hat{\theta}= \underset{\theta \in \Theta_{\cM}}{\operatorname{argmin}}\,  \left\{ \alpha \sum_{M \in \cM}
\theta_M \|z-\hat{f}_M\|^2+(1-\alpha) \Big\|z-\sum_{M \in \cM} \theta_M \hat{f}_M\Big\|^2 + \Pen(\theta) \right\}
\ee
for some $0 \leq \alpha \leq 1$ and a penalty $\Pen(\theta)$, and $\Theta_{\cM}$ is  the simplex
\be \label{eq:cone}
\Theta_{\cM}=\{\theta \in \mathbb{R}^{card(\cM)}: \theta_M \geq 0,\;
\sum_{M \in \cM} \theta_M=1\}
\ee
In particular, Dai {\em et al.} (2012) considered   $\Pen(\theta)$ proportional to the 
Kullback-Leibler divergence $KL(\theta,\pi)$ for some prior $\pi$ on
$\Theta_{\cM}$. For direct problems, they derived sharp oracle inequalities both 
in expectation and with high probability for Q-aggregation with such   penalty. In fact,
EWA can also be viewed as an extreme case of Q-aggregation
of Dai {\em et al.} for $\alpha=1$ (see Rigolette \& Tsybakov, 2012).
However, the results for Q-aggregation with Kullback-Leibler-type penalty are not valid for ill-posed problems.

In this section we propose a different type of penalty for Q-aggregation in (\ref{eq:Q_weights})
that is specifically designed for the solution of inverse problems. In particular, this penalty
allows one to obtain sharp oracle inequalities both in expectation and with high probability in both mild and
severe ill-posed linear inverse problems. Namely, we consider the penalty
\be \label{eq:Q-penalty}
\Pen(\te)=\frac{4\sigma^2\, (\delta + 1) \ln p}{ \nu^{2}_r}
\sum_{M \in \cM} \theta_M \|\Psi_M\|^2_F
\ee
with a tuning parameter $\delta>0$.
For such a penalty and $\alpha=1/2$, the resulting $\hat{\theta}$ is
\be \label{eq:opt_weights}
\hat{\theta}= \underset{\theta \in \Theta_{\cM}}{\operatorname{argmin}}\,  \left\{\sum_{M \in \cM}
\theta_M \left(\|z-\hat{f}_M\|^2+ \frac{8\sigma^2\, (\delta + 1) \ln p}{ \nu^{2}_r}\,
 \|\Psi_M\|^2_F  \right)+ \Big\|z-\sum_{M \in \cM} \theta_M\hat{f}_M\Big\|^2
\right\}
\ee
Note that the first term in the minimization criteria (\ref{eq:opt_weights}) is the same
as in model selection (\ref{eq:mhat}) with the penalty (\ref{eq:penalty}) for $a=1/2$.
The presence of the second term is inherent for Q-aggregation.  In fact, the model
selection estimator $\hat{f}_{\widehat M}$ from Section \ref{sec:model_select} is
a particular case of a Q-aggregate estimator $\hat{f}_{\hat \te}$ with the weights 
obtained  by solution of problem   \eqref{eq:Q_weights} with $\alpha =1$. 

The  non-asymptotic upper bounds for the quadratic risk
of  $\hat{f}_{\hat \te}$, both with high probability and in expectation,
are given by the following theorem~:

\begin{theorem} \label{th:oracle_aggregation}
Consider the model (\ref{eq:discrete}) and the Q-aggregate estimator
$\hat{f}_{\te}$ given by \eqref{eq:fhattheta}, where the weights ${\te}$ are selected as a solution of the optimization problem
\eqref{eq:opt_weights}. Then,

\begin{enumerate}

\item
With probability at least
$1 - \sqrt{\frac{2}{\pi}}~p^{-\delta}$
\be \label{eq:new_oraclehp}
\|\hat{f}_{\hat{\te}}-f\|^2 \leq   \min_{M \in \cM}
\left\{ \|f_M - f\|^2 + 10 \nu^{-2}_r\, \sigma^2 (\delta + 1)\ \|\Psi_M\|^2_F\ \ln p \right\}
\ee
\ignore{
\be \label{eq:new_oraclehp}
\|\hat{f}_{\hat{\te}}-f\|^2 \leq   \min_{M \in \cM}
\left\{ \|f_M - f\|^2 + \frac{5}{2}~ U_M \right\}
\ee
}

\item If, in addition, we restrict the set of admissible models to
$\cM_\gamma=\{M \in \cM:\|\Psi_M\|_F^2 \leq \gamma^2 n\}$ for some constant $\gamma$,
then
\be \label{eq:new_oraclee}
E\|\hat{f}_{\widehat M}-f\|^2  \leq  \min_{M \in \cM_\gamma}
\left\{ \|f_M - f\|^2 +    \nu^{-2}_r\, \sigma^2 \|\Psi_M\|^2_F\, \lkr 3 +  2 (\delta + 1)\, \ln p  \rkr \right\} + 2 \gamma^2\, n p^{-\delta}
\ee
\ignore{
and choose $\delta \geq 2(1 + \ln n/\ln p)$, then
\be \label{eq:new_oraclee}
E\|\hat{f}_{\widehat M}-f\|^2  \leq  \min_{M \in \cM_\gamma}
\left\{ \|f_M - f\|^2 +   U_M  \lkr 1 + \frac{\nu^2_r}{2 (\delta + 1)\, \ln p} \rkr + \frac{\gamma}{2} \right\}
\ee
}
\end{enumerate}

\end{theorem}

Observe that unlike Theorem \ref{th:oracle_mod_select} for model selection estimator, inequalities in both \eqref{eq:new_oraclehp} and \eqref{eq:new_oraclee}
for Q-aggregation are sharp.


\section{Simulation study}
\label{sec:simulations}

In this section we present   results of a simulation study that
illustrates the performance of the model selection estimator ${\widehat M}$ from (\ref{eq:mhat}) with the penalty (\ref{eq:penalty}) and
the Q-aggregation estimator $\hat{f}_{\hat \theta}$ given by \eqref{eq:fhattheta}  where the weights are defined in (\ref{eq:opt_weights}).

The data were generated w.r.t. a (discrete) ill-posed statistical linear problem (\ref{eq:discrete})  corresponding to
the convolution-type operator $Af(t)=\int_0^t e^{-(t-x)}f(x) dx,\; 0 \leq t \leq 1$,  on a regular grid $t_i=i/n$:
$$
A_{ij}=e^{-\frac{i-j}{n}}\, I(j \leq i),\;\;\;i,j=1,\ldots,n,
$$
where $I(\cdot)$ is the indicator function and $n=128$.

We considered the dictionary obtained by combining two wavelet bases of different type: 
the  Daubechies 8  wavelet basis $\{\phi^D_{3,0},...,\psi^D_{3,0},...,\psi^D_{6,63}\}$ and the 
Haar basis $\{\phi^H_{3,0},...,\psi^H_{3,0},...,\psi^H_{5,53}\}$, with 
the overall dictionary size $p=128 +64 = 192$. In our notations, $\phi^D$ and $\phi^H$ are the scaling 
functions, while $\psi^D$ and $\psi^H$ are the wavelets functions of the Daubechies and Haar bases respectively 
with the initial resolution level  $J_0=3$.

In order to investigate the behavior of the estimators, we considered various sparsity and noise levels. 
In particular,  we used four test functions, presented in Figure \ref{fig:plots}, that correspond  to different sparsity scenarios: 
\begin{enumerate}
\item $f_1=\phi^D_{3,4} + \phi^H_{3,0}$ (high sparsity)
\item $f_2 = \phi^D_{3,0} + \phi^D_{3,6} + \psi^D_{3,7} +\phi^H_{3,6}$
(moderate sparsity)
\item $f_3 =  \phi^D_{3,1} + \phi^D_{3,5} + \phi^D_{3,7} + \psi^D_{3,0} + \psi^D_{3,3} + \psi^D_{3,5} + \phi^H_{3,0} + \phi^H_{3,3}$ (low sparsity)
\item $f_4$ is the well-known HeaviSine function from Donoho \& Johnstone (1994)
(uncontrolled sparsity)
\end{enumerate}
%
%
For each test function, we used three different values of   $\sigma$ that were chosen to ensure a signal-to-noise ratios $SNR=10, 7, 5$, where 
$SNR(f)=\|f\|/\sigma$.

The accuracy of each estimator was measured by its
relative integrated error:
$$
R(\hat{f}) =  \| f - \hat{f}\|^2 / \| f \|^2
$$
Since the model selection estimator $\hat{f}_{\widehat M}$ involves minimizing a cost function of the form $-\|\hat{f}_M\|^2 +  4 \sigma^2 \lambda \ln p \|\Psi_M\|^2_F$ over the entire model space $\cM$ of a very large size,  
we used a Simulated Annealing (SA) stochastic optimization algorithm for an approximate solution.
The SA algorithm is a kind of a Metropolis sampler where the acceptance probability is ``cooled down'' by a synthetic temperature
parameter (see Br\'emaud 1999, Chapter 7, Section 8). More precisely, if $M^{(r)}$ is a  solution at step $r > 0$ of the algorithm, 
at step $r+1$ a tentative solution $M^*$ is selected according to a given symmetric proposal distribution  and it is accepted with probability
\begin{equation} \label{eq:accprob}
a(M^*,M^{(r)})= \min \left\{ 1, \exp \left(- \frac{\pi(M^*)-\pi(M^{(r)})}{T^{(r)}} \right)\right\}.
\end{equation}
where $T^{(r)}$ is a temperature parameter at step $r$.
The expression (\ref{eq:accprob}) is motivated by the fact that while $M^*$ is always accepted if $\pi(M^{(r)}) \geq \pi(M^*)$, 
it can still be accepted even if $\pi(M^{(r)}) < \pi(M^*)$ in spite of being worse than the current one. The chance of acceptance of 
$M^*$ for the same value of $\pi(M^{(r)})- \pi(M^*) <0$ diminishes at every step as the  temperature $T^{(r)}$ decreases with $r$.
The law that reduces the temperature is called the cooling schedule, in  particular, here we choose  $T^{(r)}=1/(1+ \log(r))$.

In this paper we adopted the classical symmetric uniform proposal distribution and selected a starting solution 
$M^{(0)}$ according to the following initial probability 
\begin{equation} \label{eq:birthprop}
p(j)=C \exp\{  \langle \psi_j, y \rangle^2-c \|\psi_j\|^2 \}, \quad \mbox{for } j=1,\ldots,p
\end{equation}
where $c= \sum_j \langle \psi_j, y \rangle^2 / \sum_j \|\psi_j\|^2$ and 
$C= \left[\sum_{j=1}^p \exp\{  \langle \psi_j ,y \rangle ^2-c \|\psi_j\|^2 \}\right]^{-1}$ is the normalizing constant.
Observe that the argument in the exponent in (\ref{eq:birthprop}) is the difference of
$  \langle \psi_j, y \rangle ^2 $ and $\|\psi_j\|^2$ where the first term  $\langle \psi_j ,y \rangle ^2$   is the squared $j$-th empirical coefficient 
 while the second term  $\|\psi_j\|^2$  is the increase in the variance due to the addition of the  $j$-th dictionary function.
Hence, the prior $p(j)$ is more likely to choose  dictionary functions with   small variances that are highly correlated to the true function $f$.

Thus, the adopted SA procedure can be summarized as follows:\begin{itemize}

\item generate a random number $m \leq n/\log(p)$. Set  $T^{(1)}=1$ 

\item  generate a starting solution $M^{(0)}$ with $\card(M^{(0)})=m$ by sampling indices $j \in  \{ 1, \ldots, p \}$ 
 according to the probability given by equation (\ref{eq:birthprop})


\item  repeat for $r=1,2,...r_{max}$
\begin{enumerate}

\item    generate a variable $j^* \sim\ \mbox{Uniform}\, (1,...,p)$  

\item \emph{if} $j^* \notin M^{(r)}$ propose $M^{*}=M^{(r)} \cup \{j^* \}$

      \emph{else}

               propose  $M^{*}=M^{(r)}-\{j^* \}$


\item  with probability $a(M^{*},M^{(r)})$ given in equation (\ref{eq:accprob}) assign $M^{(r+1)}=M^*$, otherwise $M^{(r+1)}=M^{(r)}$

\item update the temperature parameter  $T^{(r+1)}=1/(1+\log(r+1))$
\end{enumerate}

\end{itemize}

\noindent
While various stopping criteria could be used in the SA procedure, we found $r_{max}=100,000$ to be sufficient
for obtaining a good approximation of the global minimum in (\ref{eq:mhat}). Once the algorithm is terminated, we evaluated 
$\hat{f}_{\widehat M}$, where  $\widehat{M}=\arg \min_{0 \leq r \leq r_{max}}\pi(M^{(r)})$ was the ``best'' model in the chain of models
generated by SA algorithm.

Similarly, the Q-aggregation estimator $\hat{f}_{\hat \theta}$ involves
computationally expensive aggregation of estimators over the entire model space $\cM$.
We, therefore, approximated it by aggregating over the subset $\cM^{'}$ of the last 50 ``visited'' models in the SA chain, i.e.
$\hat{f}_{\hat \theta} = \sum_{M \in \cM^{'}} \hat{\theta}_M \hat{f}_{M}$ with $\hat{\theta}$ being a solution of (\ref{eq:opt_weights}).

\begin{figure} 
\centering
\includegraphics[width=1\textwidth]{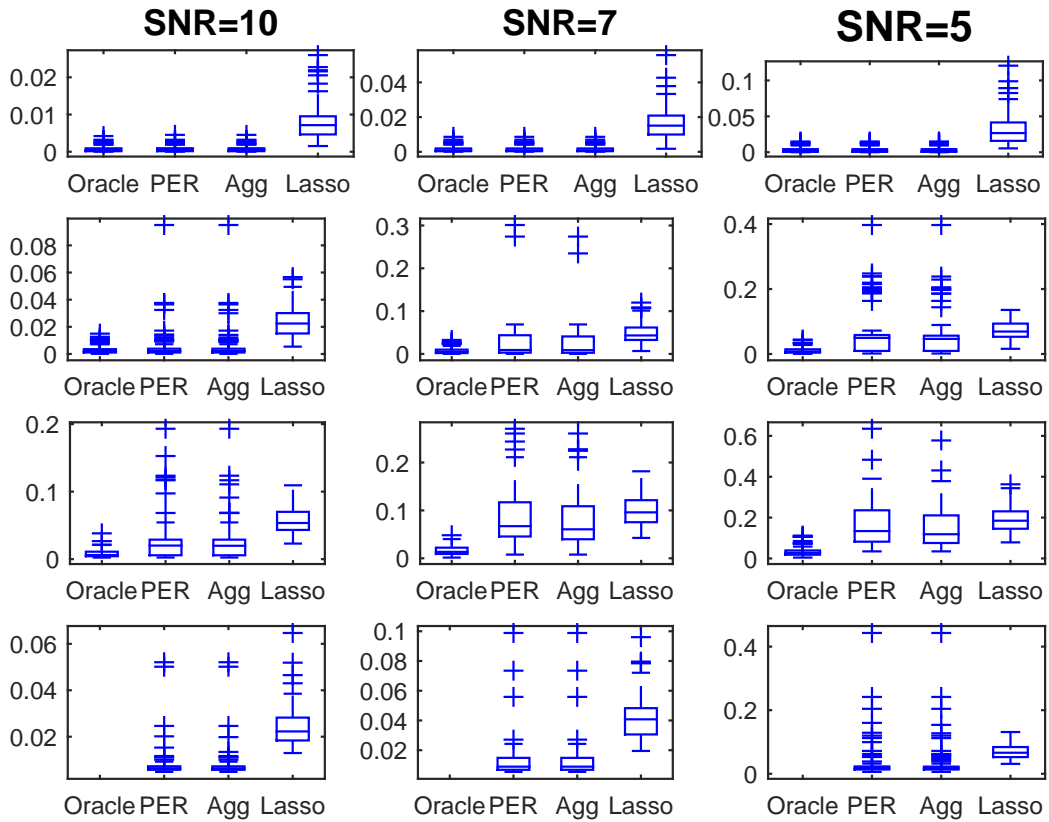}
 \caption{The   boxplots of the relative integrated errors of $\hat{f}_{oracle}$, $\hat{f}_{\widehat M}$, $\hat{f}_{\theta}$  and $\hat{f}_{Lasso}$
over 100 independent   runs. Top row: $f_1$, second row: $f_2$, third row $f_3$, bottom row $f_4$. 
}
\label{fig:boxplot}
\end{figure}

For $f_1$, $f_2$ and $f_3$ we also considered the oracle projection estimator $\hat{f}_{oracle}$ based on the true model.
In addition, we compared the proposed estimators with the Lasso-based estimator 
$\hat{f}_{Lasso}=\sum_{j=1}^p  \hat{\theta}_j  \phi_j$ of Pensky (2016), 
where the vector of coefficients $\hat{\theta}$ is a solution of the following 
optimization problem 
$$
\hat{\theta} =   \arg\min_{\theta}    \lfi     \left\| \sum_{j=1}^p  \theta_j  \phi_j \right\|^2  - 
2 \sum_{j=1}^p   \theta_j \langle   y,   \psi_j \rangle + 
\lambda \sum_{j=1}^p   |\theta_j|\,  \| \psi_j \|^2    \rfi,
$$
and $\lambda$ is a tuning parameter.

The tuning parameters $\lambda$ for $\hat{f}_{\widehat M}$ and $\hat{f}_{Lasso}$,
were chosen by minimizing the error on a grid of possible values. To
reduce heavy computational costs we used the same $\lambda$ of
$\hat{f}_{\widehat M}$ for all 50 aggregated models used for calculating $\hat{f}_{\hat \theta}$.

\begin{figure} 
\centering
\includegraphics[width=0.45\linewidth]{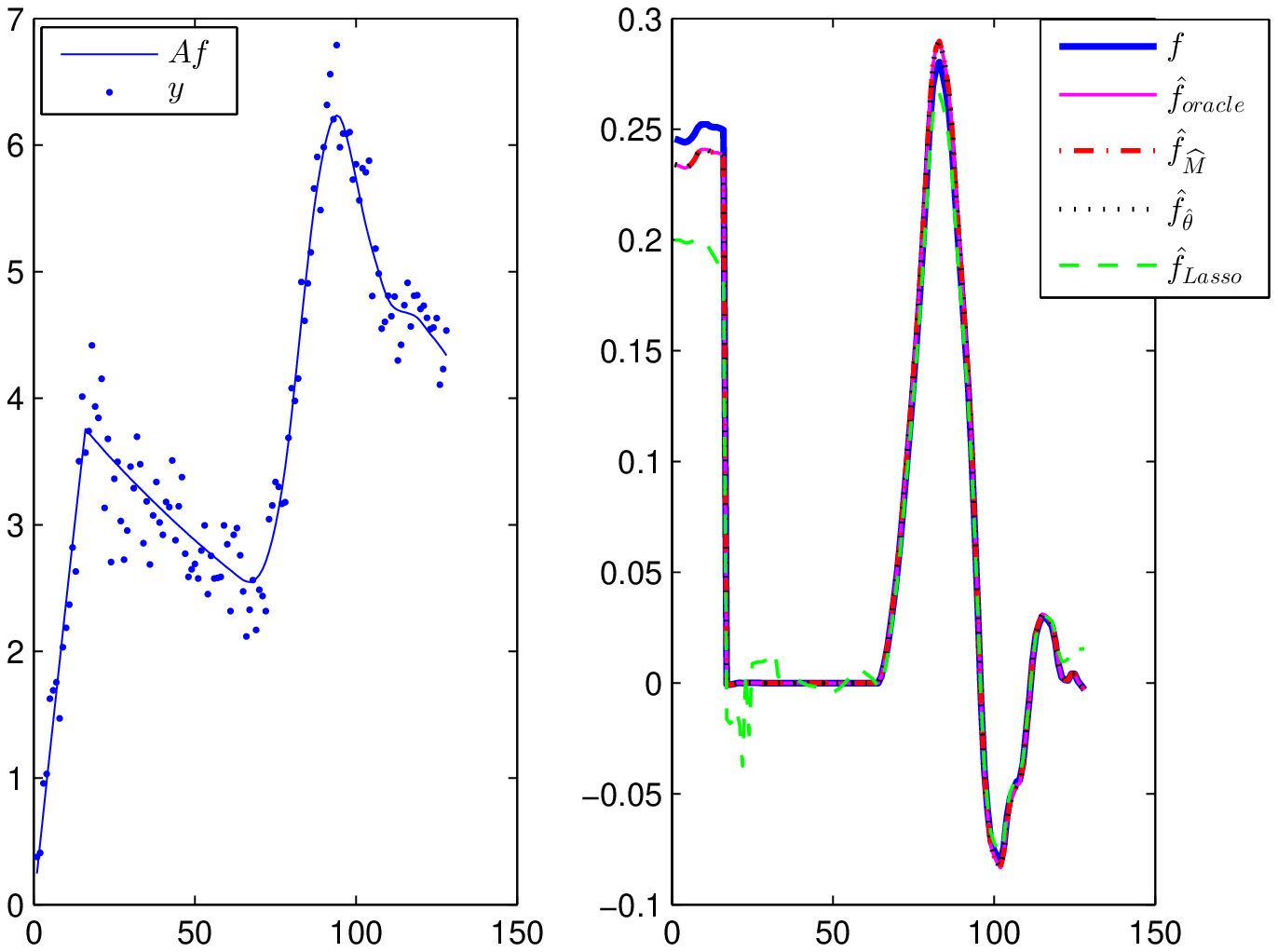}
\includegraphics[width=0.45\linewidth]{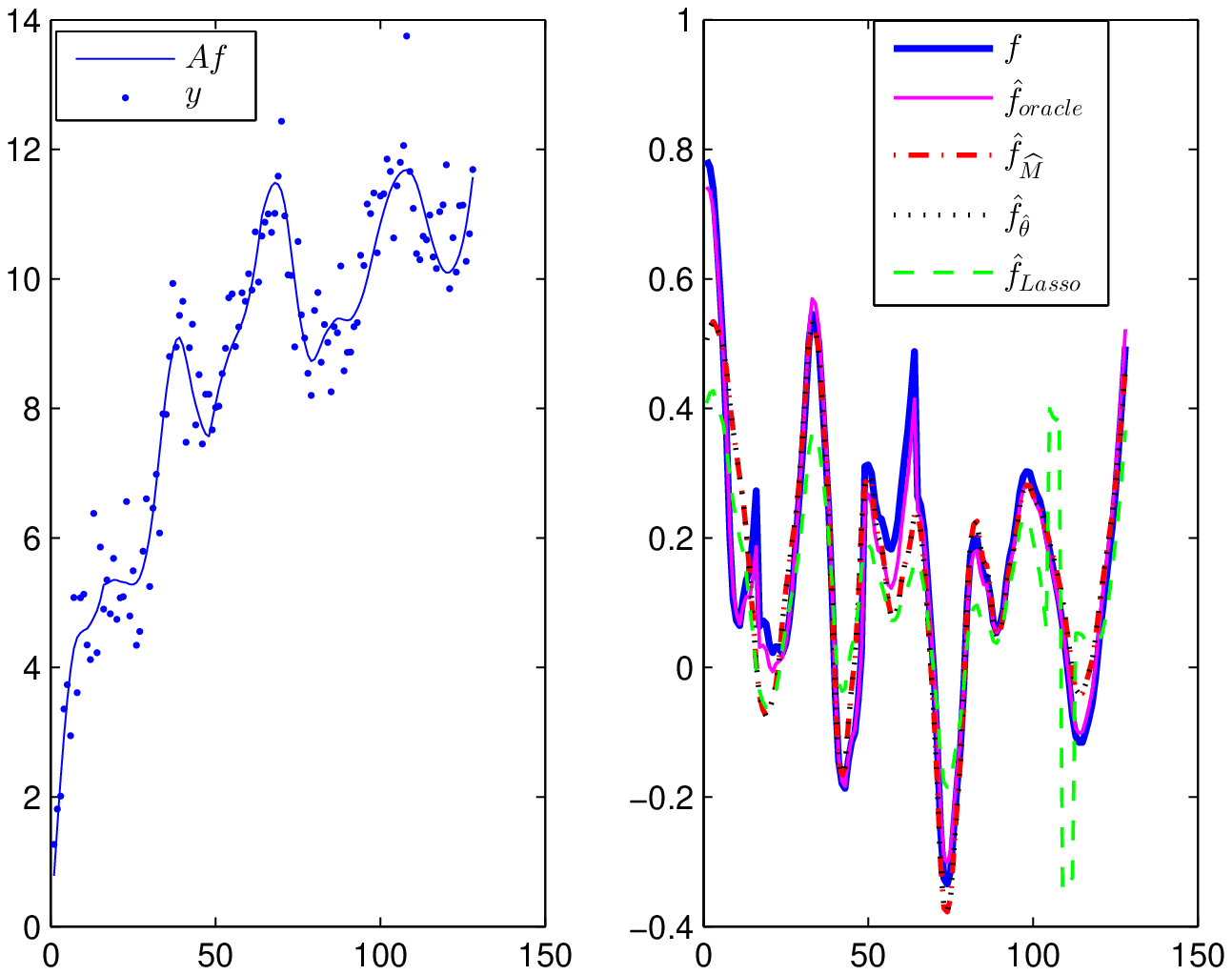}\\
\includegraphics[width=0.45\linewidth]{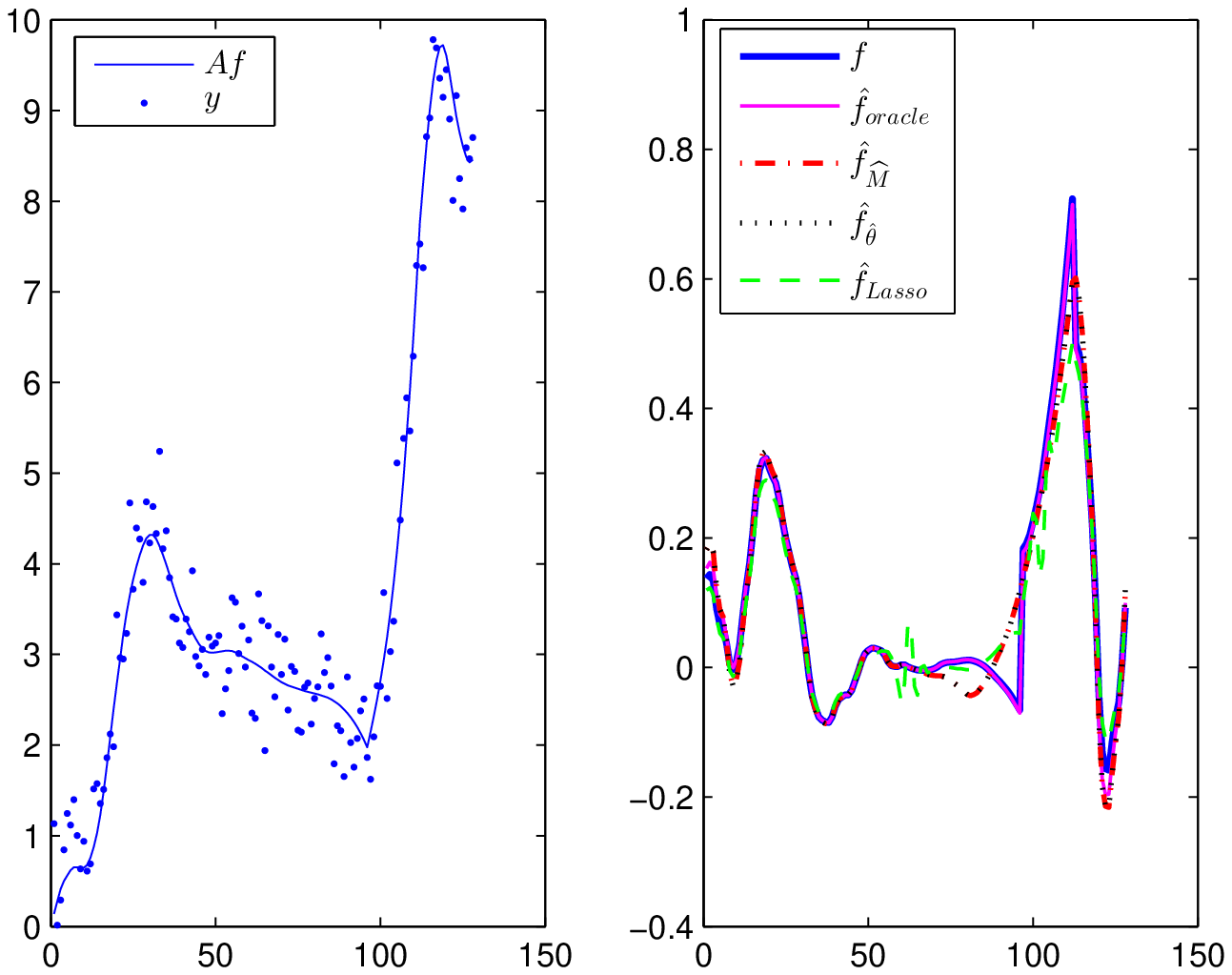}
\includegraphics[width=0.45\linewidth]{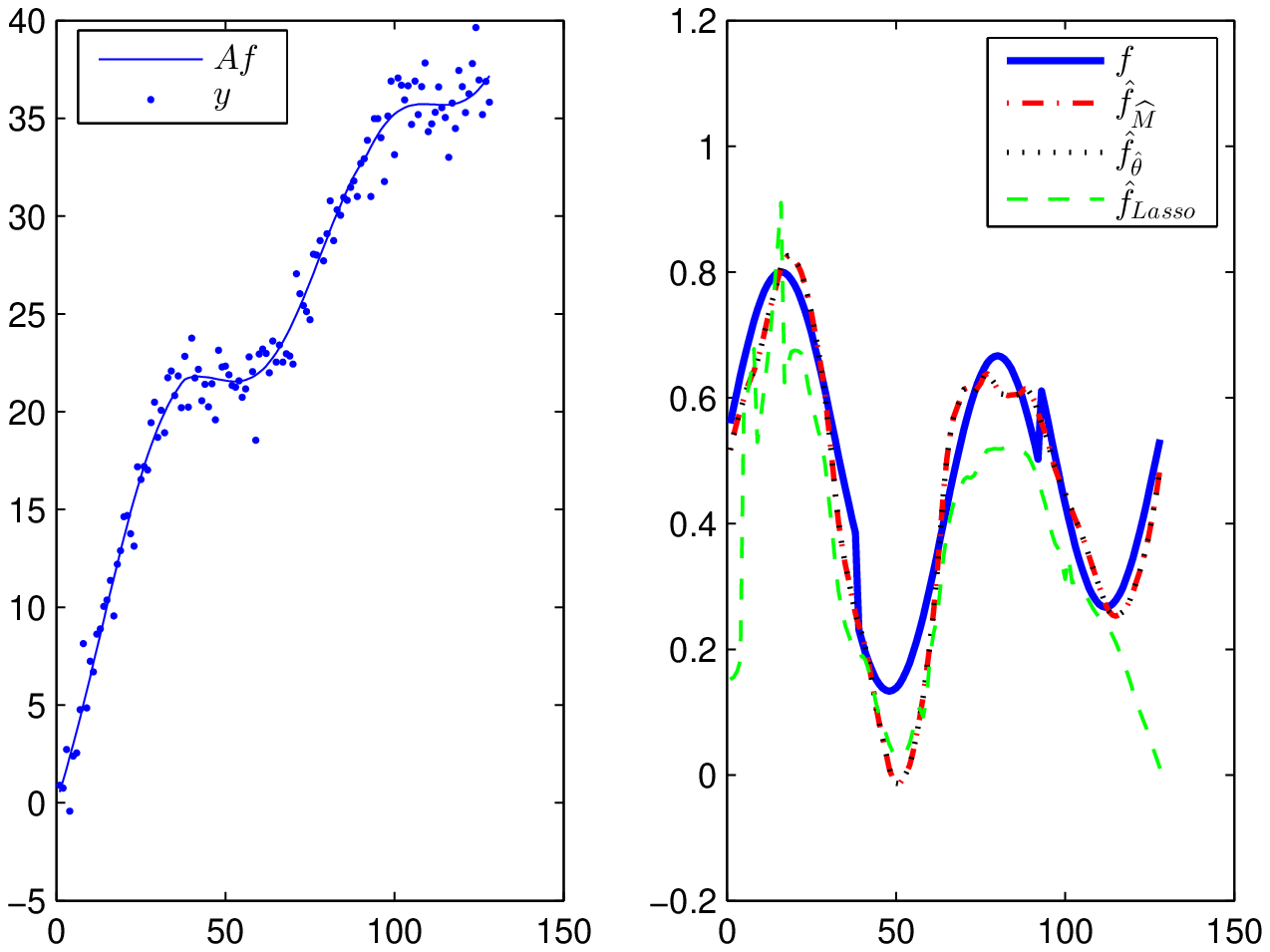}
\caption{The test functions, observed data and estimators for $SNR=5$. In each block of figures corresponding to four test functions, the left panel shows 
the data $y$ and the true $Af$; the right panel shows the true $f$ and
the four estimators. 
Top left: $f_1$; top right: $f_2$; bottom left: $f_3$; bottom right: $f_4$. }
\label{fig:plots}
\end{figure}

Figure \ref{fig:boxplot} presents the boxplots of $R(\hat{f})$ over 100 independent runs 
for  $\hat{f}_{oracle}$  (for $f_1, f_2$ and $f_3$), $\hat{f}_{\widehat M}$, $\hat{f}_{\hat \theta}$ and $\hat{f}_{Lasso}$.
Performances of all estimators  deteriorate as SNR decreases especially for the less sparse test functions. 
The  estimators $\hat{f}_{\widehat M}$ and $\hat{f}_{\hat \theta}$ always outperform $\hat{f}_{Lasso}$ 
and, as it is expected from our theoretical statements, $\hat{f}_{\hat \theta}$ yields somewhat better results than $\hat{f}_{\widehat M}$. 
We expect that the differences in precisions of $\hat{f}_{\widehat M}$ and $\hat{f}_{\hat \theta}$ would be more
significant  if we carried out aggregation over a larger portion of the model space than the last 50 visited models. 
Figure \ref{fig:plots} illustrates these conclusion by displaying examples of the estimators for $SNR=5$.
We should also mention that  estimator  $\hat{f}_{\widehat M}$  was usually  more sparse    than $\hat{f}_{Lasso}$.

\ignore{
The oracle estimator $\hat{f}_{oracle}$ and the model selection estimator
$\hat{f}_{\widehat M}$ were hardly distinguishable from  $\hat{f}_{\hat \theta}$ on the plot and were, therefore, omitted for the clarity of exposition.
}

\ignore{ 
\begin{figure} 
\centering
\includegraphics[width=0.45\linewidth]{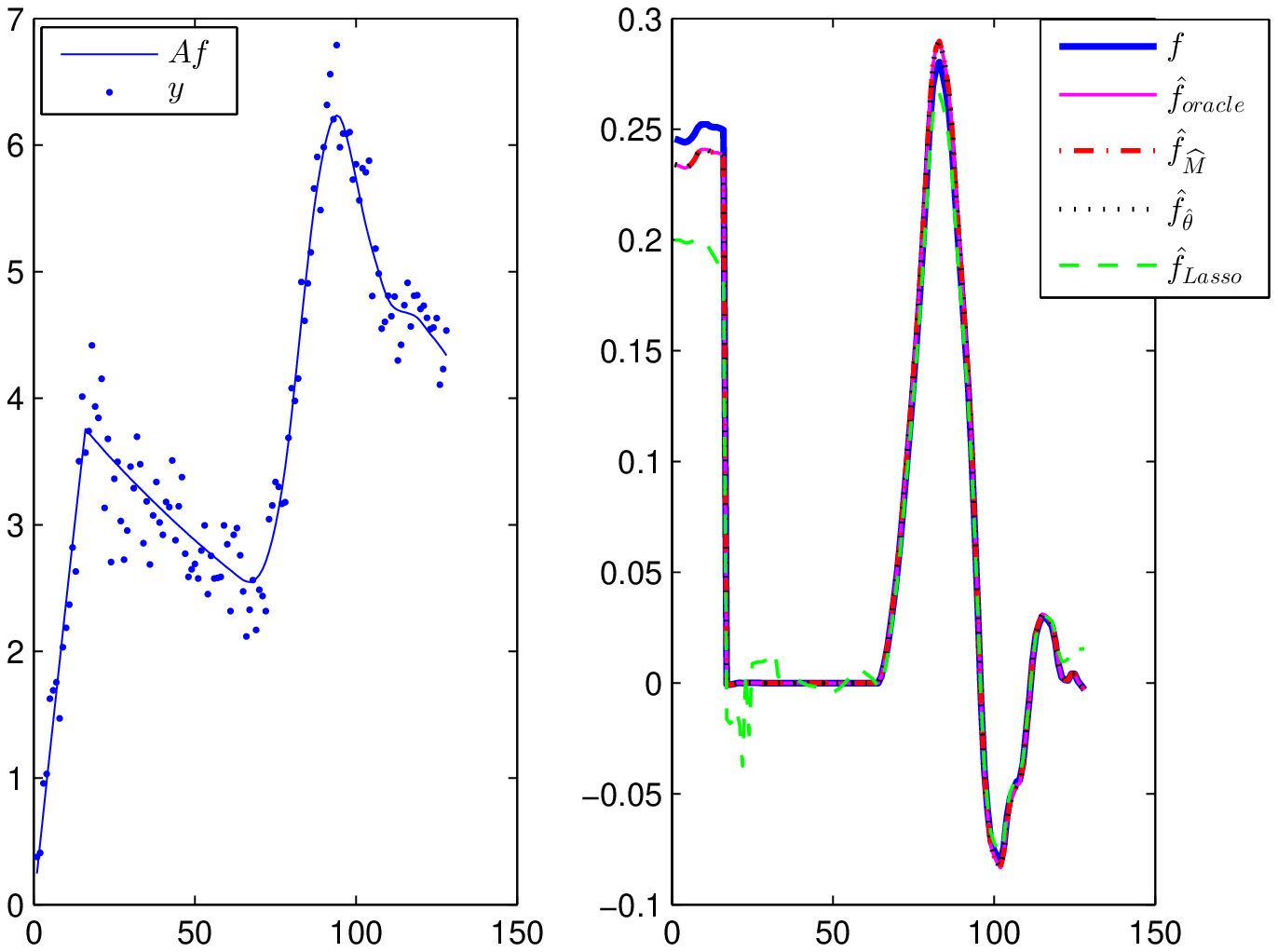}
\includegraphics[width=0.45\linewidth]{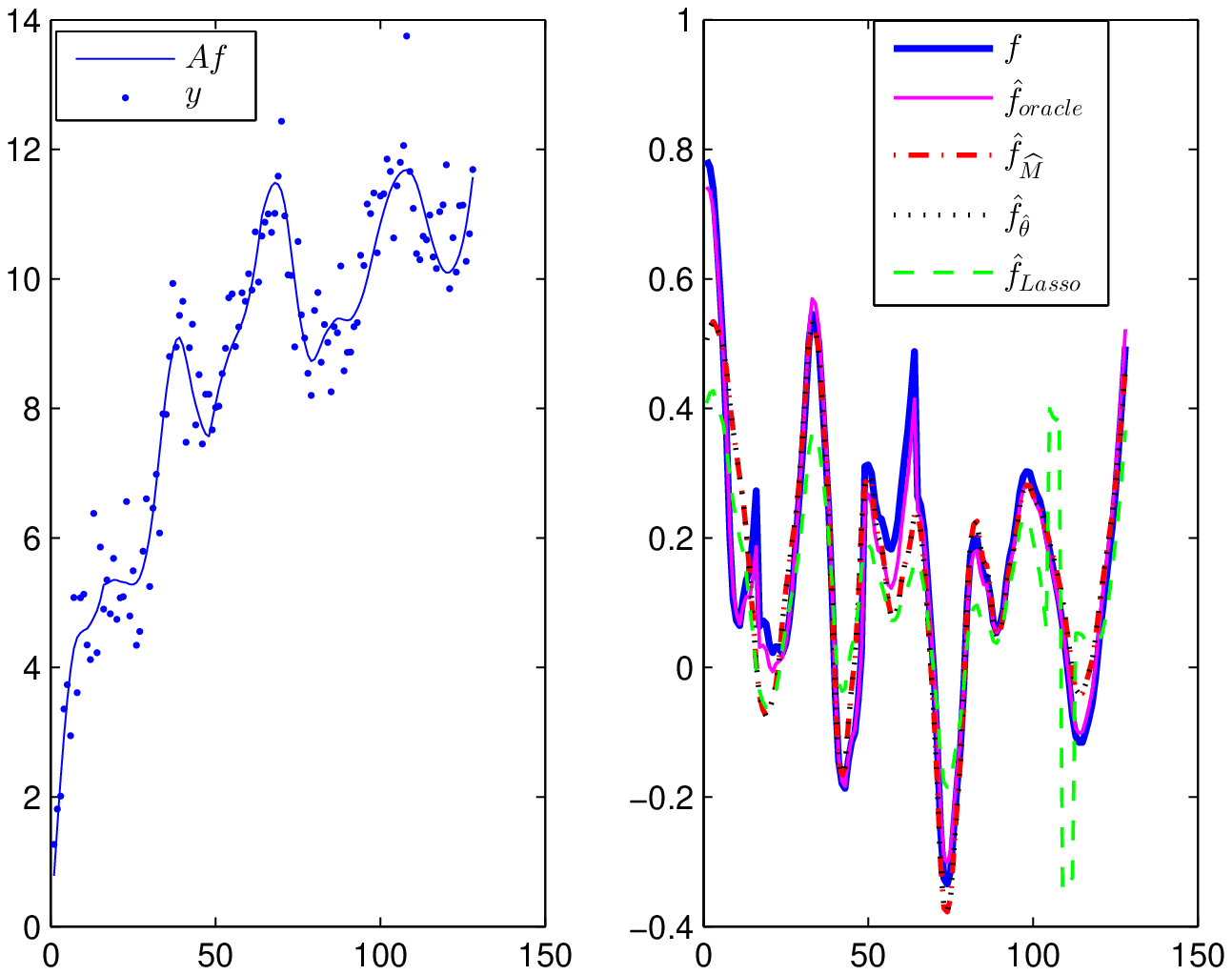}\\
\includegraphics[width=0.45\linewidth]{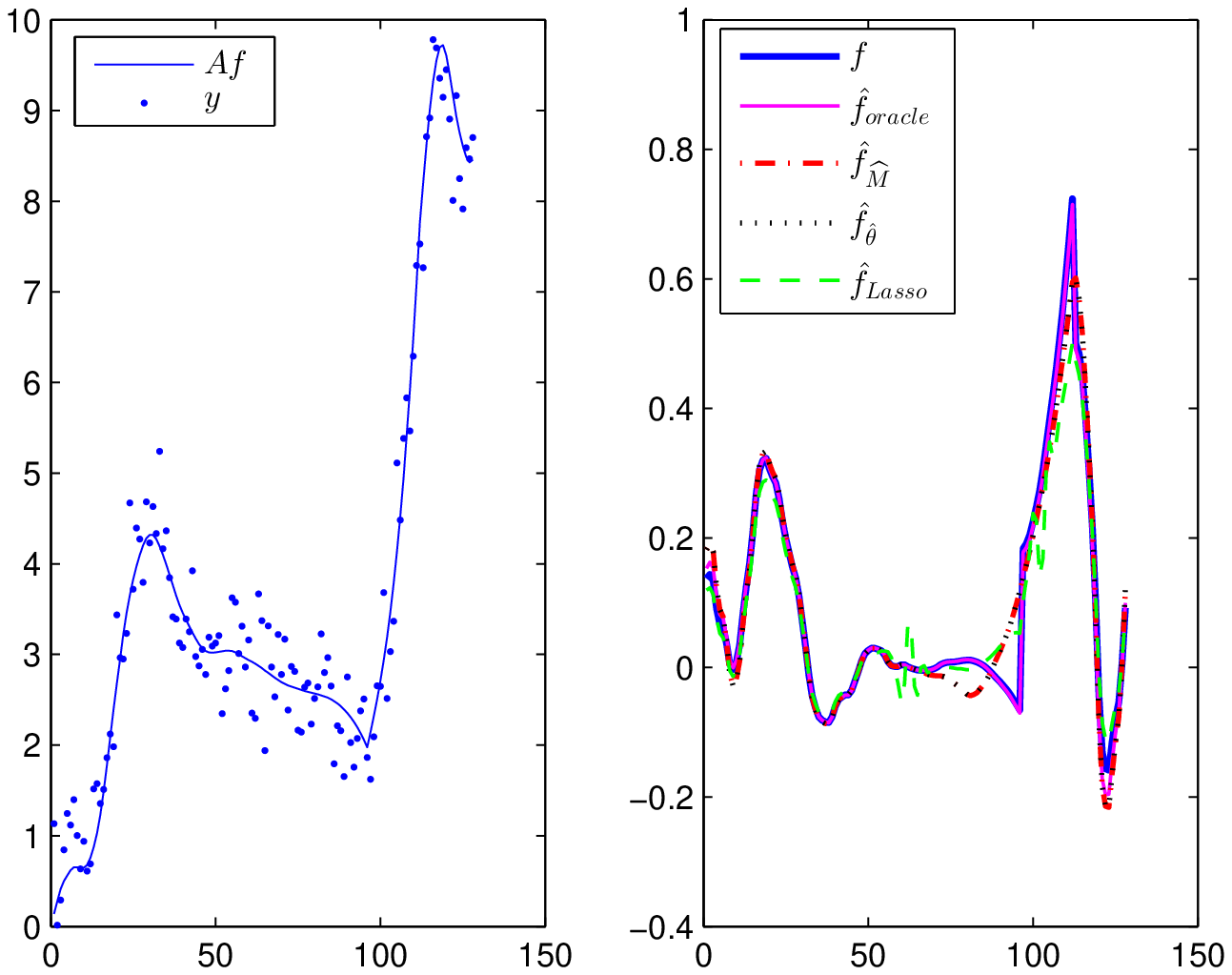}
\includegraphics[width=0.45\linewidth]{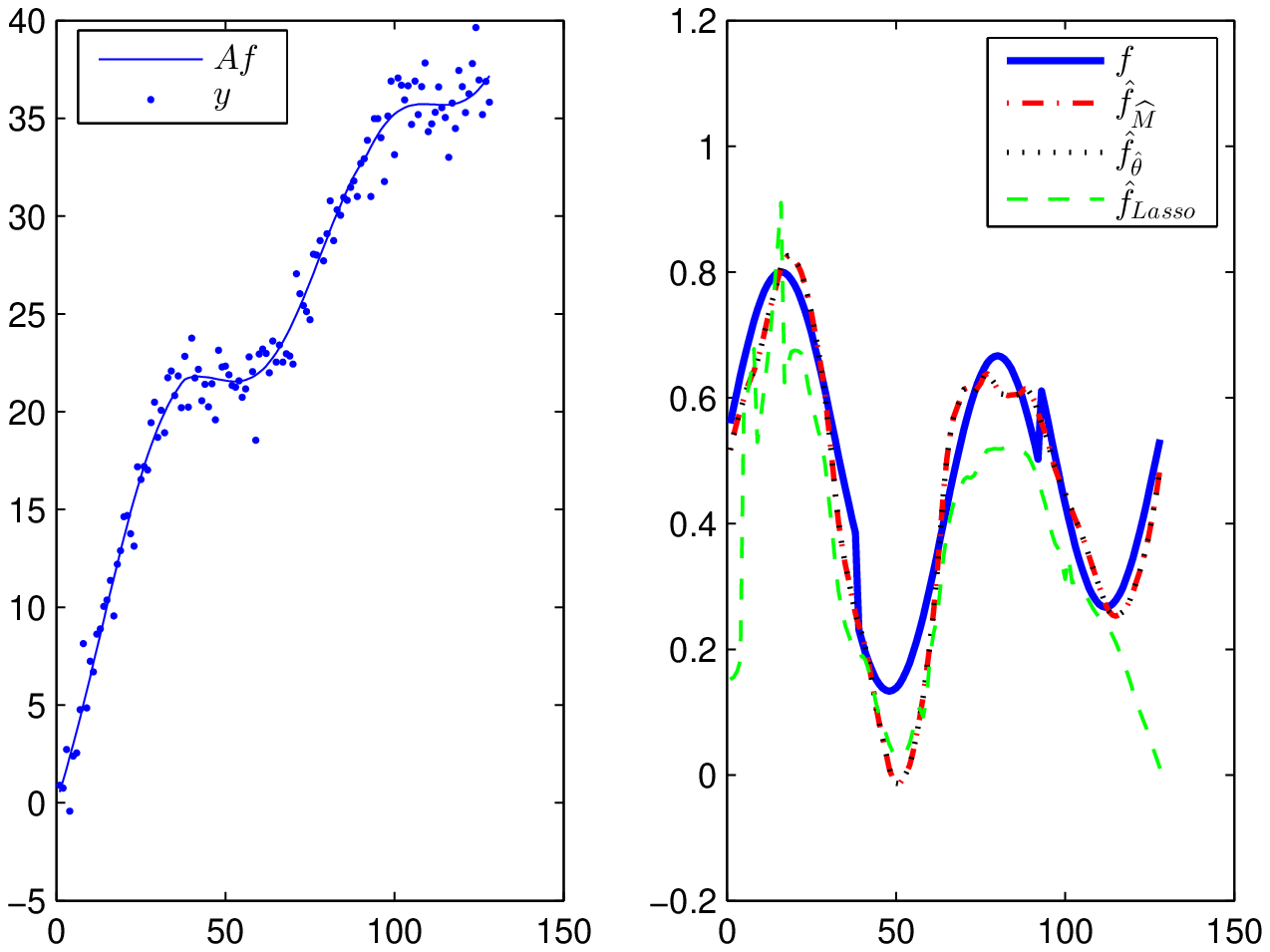}
\caption{The test functions, observed data and estimators for $SNR=5$. In each block of figures corresponding to four test functions, the left panel shows 
the data $y$ and the true $Af$; the right panel shows the true $f$ and
the four estimators. 
Top left: $f_1$; top right: $f_2$; bottom left: $f_3$; bottom right: $f_4$. }
\label{fig:plots}
\end{figure}
}

\section*{Acknowledgments}

Marianna Pensky   was  partially supported by National Science Foundation
(NSF), grants   DMS-1407475 and DMS-1712977.



\section*{Appendix}
The proofs of the main results are based on the following auxiliary lemmas.

\begin{lemma} \label{lem:lemma1}
For any $x>0$,
$$
P\left( \sup_{M \in \cM} \left\{\|\Psi_M^T \varepsilon\|^2 -2 \sigma^2 \|\Psi_M\|^2_F
(\ln p + x) \right\} \leq 0 \right) \geq 1-{\sqrt \frac{2}{\pi}}~ e^{-x}
$$
\end{lemma}
\begin{proof}
For any model $M \in \cM$,
$\|\Psi_M^T \varepsilon\|^2=\sum_{j \in M} (\psi_j^T \varepsilon)^2$, where
$\psi_j^T \varepsilon \sim N(0, \sigma^2 \|\psi_j\|^2)$. By Mill's ratio
$$
P\left((\psi_j^T \varepsilon)^2 > 2 \sigma^2 \|\psi_j\|^2(\ln p + x)\right)
\leq {\sqrt \frac{2}{\pi}}~ p^{-1} e^{-x}
$$
for any $x>0$. Then,
\be \nonumber
\begin{split}
P&\left( \bigcap_{j \in M}\left\{(\psi_j^T \varepsilon)^2 - 2 \sigma^2 \|\psi_j\|^2(\ln p + x)
\leq  0\right\}\right) \geq 1-\sum_{j \in \cM} P\left((\psi_j^T \varepsilon)^2 - 2 \sigma^2 \ \|\psi_j\|^2(\ln p + x)>0\right) \\
& \geq 1- {\sqrt \frac{2}{\pi}}~ e^{-x}
\end{split}
\ee
and, therefore,
\be \nonumber
\begin{split}
P&\left( \sup_{M \in \cM} \left\{\|\Psi_M^T \varepsilon\|^2 -2 \sigma^2 \|\Psi_M\|^2_F
(\ln p + x)  \leq 0 \right\} \right) \geq
P\left(\bigcap_{M \in \cM} \bigcap_{j \in M} \left\{(\psi_j^T \varepsilon)^2 -
2 \sigma^2 \ \|\psi_j\|^2(\ln p + x) \leq 0\right\}\right) \\
& =
P\left(\bigcap_{j=1}^p\left\{(\psi_j^T \varepsilon)^2 - 2 \sigma^2 \ \|\psi_j\|^2(\ln p + x)
\leq 0\right\}\right) \geq 1-{\sqrt \frac{2}{\pi}}~ p^{-1} e^{-x}
\end{split}
\ee
\end{proof}


\begin{lemma} \label{lem:lemma2}
For any $M_1, M_2 \in \cM$, one has
\beqn
\Tr(\Psi_{M_1 \cup M_2}^T \Psi_{M_1 \cup M_2}) & \leq & \Tr(\Psi_{M_1}^T \Psi_{M_1}) + \Tr(\Psi_{M_2}^T \Psi_{M_2})
\label{lem2:ineq1}\\
\lambda_{\max} (\Psi_{M_1 \cup M_2}^T \Psi_{M_1 \cup M_2}) & \leq &
2 \lkv \lambda_{\max} (\Psi_{M_1}^T \Psi_{M_1}) + \lambda_{\max} (\Psi_{M_2}^T\Psi_{M_2}) \rkv
\label{lem2:ineq2}
\eeqn
\end{lemma}
\begin{proof}
Inequality \eqref{lem2:ineq1} follows from
$$
\Tr(\Psi_{M_1 \cup M_2}^T \Psi_{M_1 \cup M_2})  =  \sum_{j \in M_1 \cup M_2} \|\psi_j\|^2
\leq  \Tr(\Psi_{M_1}^T \Psi_{M_1}) + \Tr(\Psi_{M_2}^T \Psi_{M_2})
$$
In order to prove   inequality \eqref{lem2:ineq2}, observe that
\beqns
\lambda_{\max} (\Psi_{M_1 \cup M_2}^T \Psi_{M_1 \cup M_2})
& = & \max_{x \in R^{|M_1 + M_2|}}
\frac{\| \Psi_{M_1} x_{M_1} + \Psi_{M_2 \setminus M_1} x_{M_2 \setminus M_1} \|^2}{\|x\|^2}\\
& \leq & 2\, \max_{x \in R^{|M_1 + M_2|}} \lkv \frac{\| \Psi_{M_1} x_{M_1} \|^2}{\|x\|^2}
+ \frac{\| \Psi_{M_2} x_{M_2} \|^2}{\|x\|^2} \rkv
\eeqns
which implies \eqref{lem2:ineq2}.
\end{proof}


\begin{lemma} \label{lem:lemma3}
If $\te \in \Theta_{\cM}$ where $\Theta_{\cM}$ is defined in \eqref{eq:cone}, then, for    any function $\tilde{f}$
\be \label{eq:lem3}
\sum_{M \in \cM} \hatte_M \|\tilde{f} - \hatf_M\|^2 =
\|\tilde{f} - \hatf_{\hatte}\|^2 + \sum_{M \in \cM} \hatte_M \|\hatf_{\hatte} - \hatf_M\|^2.
\ee
\end{lemma}
\begin{proof}
Note that for any $\tilde{f}$ one has
\beqns
\sum_{M \in \cM} \hatte_M \|\tilde{f} - \hatf_M\|^2  & = & \sum_{M \in \cM} \hatte_M  \|\tilde{f} - \hatf_{\hatte}\|^2 +
\sum_{M \in \cM} \hatte_M  \|\hatf_M - \hatf_{\hatte}\|^2 \\
& + & 2 \langle \tilde{f} - \hatf_{\hatte}, \sum_{M \in \cM} \hatte_M (\hatf_{\hatte} - \hatf_M \rangle.
\eeqns
Now validity of \eqref{eq:lem3} follows from the fact that $\te \in \Theta_{\cM}$, so that the scalar product term in the last identity
is equal to zero.
\end{proof}


\begin{lemma} \label{lem:lemma4}
Let $\xi$ be defined in formula \eqref{eq:z_def}  and $M_0 \subseteq \{ 1, \cdots, p\}$ be an arbitrary fixed model.  
Denote $U_M =  4 \nu^{-2}_r\, \sigma^2 (\delta + 1)  \ \|\Psi_M\|^2_F\ \ln p$ and
\be \label{eq:Delta0}
\Delta_0 = 2 \sigma \langle \xi, \hatf_{\hatte} - \hatf_{M_0} \rangle + \frac{\sigma^2}{\nu_r^2}\, \|\Psi_{M_0}^T \eps\|^2
- \frac{3}{2}\, U_{M_0} - \sum_{M \in \cM} \hatte_M U_M - \frac{1}{2}\, \sum_{M \in \cM} \hatte_M \|\hatf_M -\hatf_{M_0}\|^2.
\ee
Then, with probability at least $1 - \sqrt{2/\pi}\, p^{-\delta}$, one has $\Delta_0 \leq 0$. Moreover,
if the set of models is restricted to $\cM_\gamma=\{M \in \cM:\|\Psi_M\|_F^2 \leq \gamma^2 n\}$, then
\be \label{eq:lem4}
E (\Delta_0) \leq U_{M_0} \lkv \frac{3}{4(\delta+1) \ln p}- \frac{3}{2} \rkv + \frac{4 \sigma^2 \gamma^2}{\nu_r^2} \, n p^{-\delta}.
\ee
\end{lemma}
\begin{proof}
Note that due to   $\hatte \in \Theta_{\cM}$ , one has
\beqns
 2 \sigma \langle \xi, \hatf_{\hatte} - \hatf_{M_0} \rangle  & = &
2 \sigma \, \sum_{M \in \cM} \hatte_M \langle \xi, \hatf_{M} - \hatf_{M_0} \rangle.
\eeqns
Recall that  $\hatf_M = H_M z$, hence, applying the  inequality $2\sqrt{u v} \leq au +\frac{1}{a}v$
for any positive $a, u, v>0$ and $M \in \cM$, one obtains
\beqns
2 \sigma \,  |\langle \xi, \hatf_{\hatte} - \hatf_{M_0} \rangle| & = &
2 \sigma \, |\xi^T H_{M \cup M_0}(H_M - H_{M_0})z|\\
& \leq & 2 \sigma^2 \nu_r^{-2} \|\Psi_{M \cup M_0}^T \eps\|^2 + 0.5\, \|\hatf_M - \hatf_{M_0}\|^2.
\eeqns
Combining the last two inequalities, taking into account that
$\|\Psi_{M \cup M_0}^T \eps\|^2 \leq \|\Psi_{M}^T \eps\|^2 + \|\Psi_{M_0}^T \eps\|^2$
 and plugging them into \eqref{eq:Delta0}, obtain
\be \label{eq:lem4.1}
\Delta_0 \leq \frac{2 \sigma^2}{\nu_r^2} \, \sum_{M \in \cM} \hatte_M \|\Psi_{M}^T \eps\|^2 +
\frac{3 \sigma^2}{\nu_r^2} \, \|\Psi_{M_0}^T \eps\|^2 - \frac{3}{2} U_{M_0} - \sum_{M \in \cM} \hatte_M U_M.
\ee
Applying Lemma~\ref{lem:lemma1} with $x =  \delta\,\ln p$, expressing $U_M$ via $\|\Psi_M\|^2_F$
and taking into account that $\hatte \in \Theta_{\cM}$, we derive that $\Delta_0  \leq 0$
on the set $\Omega$ with $P(\Omega) \geq 1- \sqrt{2/\pi}~ p^{-\delta}$ on which
$$
\sup_{M \in \cM} \left\{\|\Psi_M^T \varepsilon\|^2 -2 \sigma^2 \|\Psi_M\|^2_F
 \ln p (1 + \delta) \right\} \leq 0.
$$
In order to derive inequality \eqref{eq:lem4}, plugging in the expression for $U_M$ and $U_{M_0}$ into \eqref{eq:lem4.1},
derive that $\Delta_0 \leq \Delta_{01} + \Delta_{02}$ where
\beqns
\Delta_{01} & = & \frac{3 \sigma^2}{\nu_r^2} \lkv \| \Psi_{M_0}^T \eps\|^2 - 2(\delta+1) \ln p \| \Psi_{M_0}\|^2_F \rkv \\
\Delta_{02} & = &   \frac{2 \sigma^2}{\nu_r^2}\, \sum_{M \in \cM} \hatte_M  \lkv \| \Psi_{M}^T \eps\|^2 - 2(\delta+1) \ln p \| \Psi_{M}\|^2_F \rkv
\eeqns
By direct calculations, one can easily show that
\be \label{eq:lem4.2}
E (\Delta_{01}) = U_{M_0} \lkr \frac{3}{4 (\delta+1) \ln p} - \frac{3}{2} \rkr.
\ee
In order to derive an upper bound for $E (\Delta_{02})$, recall that
\beqns
E (\Delta_{02}) & \leq & E (\Delta_{02})_{+} = \int_0^\infty P(\Delta_{02}>t) dt
\eeqns
Plugging in the expression for $\Delta_{02}$, taking into account that $\|\Psi_M\|_F^2 \leq \gamma^2 n$,
replacing $t$ by $4 \sigma^2 \gamma^2 \nu_r^{-2} n\,  \ln p\ u$,
making a change of variables for integration and applying Lemma~\ref{lem:lemma1},
similarly to the proof of Theorem~\ref{th:oracle_mod_select},  we arrive at
\be \label{eq:lem4.3}
E (\Delta_{02}) \leq  4 \sigma^2 \gamma^2 \nu_r^{-2} n p^{-\delta}
\ee
Combination of \eqref{eq:lem4.2} and \eqref{eq:lem4.3} complete the proof.
\end{proof}


\subsection*{Proof of Theorem \ref{th:oracle_mod_select}}
%
Let $z$ be defined in \eqref{eq:z_def}.
Since $\widehat{M}$ is the minimizer in (\ref{eq:mhat}), for any given model
$M \in \cM$
$$
\|z-\hat{f}_{\widehat M}\|^2 +\Pen(\widehat{M}) \leq \|z-\hat{f}_M\|^2 +\Pen(M)
$$
and, by a straightforward calculus, one can easily verify that
\be \label{eq:a1}
\|\hat{f}_{\widehat M}-f\|^2 \leq
\|\hat{f}_M-f\|^2+2\langle \xi, \hat{f}_{\widehat M}-\hat{f}_M\rangle+\Pen(M)-\Pen(\widehat{M})
\ee
Denote $\widetilde{M}=\widehat{M} \cup M$ and recall that, by the definition of $\cM$,
$|\widetilde{M}| \leq r$.
Thus, by the Cauchy-Schwarz inequality and the definition of $\nu^2_r$
\be
\begin{split} \label{eq:a2}
2\langle \xi, \hat{f}_{\widehat M}-\hat{f}_M\rangle & =~2\langle (A^T A)^{-1}A^T\varepsilon,
\Phi_{\widetilde M}(\hat{\theta}_{\widehat M}-\hat{\theta}_M)\rangle~=~
2\langle \Psi_{\widetilde M}^T \varepsilon, \hat{\theta}_{\widehat M}-\hat{\theta}_M\rangle \\
& ~\leq 2 \|\Psi_{\widetilde M}^T \varepsilon\|~ \|\hat{\theta}_{\widehat M}-\hat{\theta}_M\|
~\leq 2~ \|\Psi_{\widetilde M}^T \varepsilon\|~ \nu_r^{-1}\, \|\hat{f}_{\widehat M}-\hat{f}_M\|
\end{split}
\ee
Using inequalities $2\sqrt{u v} \leq au +\frac{1}{a}v$ for any positive $a, u, v>0$
and $\|u-v\|^2 \leq 2(\|u\|^2+\|v\|^2)$, from (\ref{eq:a2}) obtain
$$
2\langle \xi, \hat{f}_{\widehat M}-\hat{f}_M\rangle ~\leq ~
a\|\hat{f}_{\widehat M}-f\|^2 + a \|\hat{f}_M-f\|^2+\frac{2}{a\nu_r} \|\Psi_{\widetilde M}^T \varepsilon\|^2
$$
Thus, from (\ref{eq:a1}) it follows that for any $0 < a < 1$
\be \label{eq:a3}
(1-a)\|\hat{f}_{\widehat M}-f\|^2 \leq (1+a)\|\hat{f}_M-f\|^2+\Pen(M)+
\frac{2}{a\nu_r^2} \|\Psi_{\widetilde M}^T \varepsilon\|^2 - \Pen({\widehat M})
\ee
Note that
$ 
\|\hat{f}_M-f\|^2=\|\hat{f}_M-f_M\|^2+\|f_M-f\|^2=
\varepsilon^T \Psi_M (\Phi_M^T \Phi_M)^{-1}\Psi_M^T\varepsilon+\|f_M-f\|^2
\leq \frac{1}{\nu_r^2} \|\Psi_M^T\varepsilon\|^2+\|f_M-f\|^2,
$ 
while, by Lemma~\ref{lem:lemma2},
$ 
\|\Psi_{\widetilde M}^T \varepsilon\|^2 ~\leq~ \|\Psi_{\widehat M}^T \varepsilon\|^2
+\|\Psi_M^T \varepsilon\|^2.
$ 
Therefore, (\ref{eq:a3}) implies
\be \label{eq:a4}
(1-a)\|\hat{f}_{\widehat M}-f\|^2 \leq (1+a)\|f_M-f\|^2+Pen(M)+
\frac{a^2+a+2}{a\nu^2_r} \|\Psi_M^T\varepsilon\|^2+
\frac{2}{a\nu^2_r}\|\Psi_{\widehat M}^T\varepsilon\|^2-Pen({\widehat M})
\ee

Applying Lemma \ref{lem:lemma1} with $x=\delta \ln p$ for any $\delta>0$,  w.p.
at least $1-{\sqrt \frac{2}{\pi}}~p^{-\delta}$ one derives
$\|\Psi_M^T\varepsilon\|^2 \leq 2\sigma^2\|\Psi_M\|^2_F
(\delta+1)\ln p$ and $\|\Psi_{\widehat M}^T\varepsilon\|^2 \leq
2\sigma^2\|\Psi_{\widehat M}\|^2_F (\delta+1)\ln p$.
Hence, for the penalty $\Pen(M)$ in (\ref{eq:penalty}),
w.p. at least $1-{\sqrt{2/\pi}~p^{-\delta}}$ one has
$\frac{2}{a\nu^2_r}\|\Psi_{\widehat M}^T\varepsilon\|^2 - \Pen({\widehat M}) \leq 0$.
Thus, after a straightforward calculus,
$$
(1-a)\|\hat{f}_{\widehat M}-f\|^2 \leq (1+a)\|f_M-f\|^2+
\frac{(a^2+a+4)}{2}~ Pen(M)
$$
for any model $M \in \cM$ that proves (\ref{eq:oraclehp}).

To prove the oracle inequality in expectation (\ref{eq:oraclee}), note that
it follows from inequality \eqref{eq:a3} and Lemma~\ref{lem:lemma2} that
\be \label{eq:exp1}
(1-a) \|\hat{f}_{\widehat M}-f\|^2 \leq (1+a)\|{f}_M-f\|^2+  (1+a)\|\Psi_{M}^T \varepsilon\|^2 + \Pen(M)+
2/(a\nu_r^2)\, \lkr \|\Psi_{M}^T \varepsilon\|^2 +   \Delta \rkr
\ee
where $\Delta = \|\Psi_{\widehat{M}}^T \varepsilon\|^2- 2(\delta+1) \ln p\, \|\Psi_{\widehat{M}}^T\|^2_F$.
Using Lemma~\ref{lem:lemma1}, obtain
\begin{align*}
E \Delta & \leq E (\Delta)_+ = \int_0^\infty P(\Delta >t)dt = 2 \gamma^2 n \, \ln p\ \int_0^\infty P(\Delta > 2 \gamma^2 n \, \ln p\ u) du\\
         & \leq  2 \gamma^2 n \, p^{-\delta} \ \int_0^\infty p^{-u} \ln p\ du =  2 \gamma^2 n \, p^{-\delta}
\end{align*}
Taking expectation in \eqref{eq:exp1} and combining it with the last inequality, we obtain (\ref{eq:oraclee}).
\newline
$\Box$


\subsection*{Proof of Theorem \ref{th:oracle_aggregation}}
Denote
\bes 
\hatS (\theta) = \frac{1}{2} \sum_{M \in \cM} \te_M \|z  - \hat{f}_M\|^2 + \frac{1}{2} \|z - \hat{f}_{\te}\|^2, \quad
S(\te) = \frac{1}{2} \sum_{M \in \cM} \te_M \|f  - \hat{f}_M\|^2 + \frac{1}{2} \|f - \hat{f}_{\te}\|^2
\ees
Note that, by definition,
\be \label{eq:S_rel_main}
S(\hatte) -  S(\te) = \frac{1}{2} \sum_{M \in \cM} (\hatte_M - \te_M) \|f  - \hat{f}_M\|^2 +
\frac{1}{2} \lkv  \|f - \hatf_{\hatte}\|^2 - \|f - \hatf_{\te}\|^2 \rkv,
\ee
and also
\be \label{eq:S_rel}
\hatS (\theta) = S(\te) + \|z\|^2 - \|f\|^2 - 2 \sigma\, \langle \xi, \hat{f}_{\te} \rangle
\ee
where $\xi$ is defined in \eqref{eq:z_def}.
By definition of $\hatte$, one has
\be \label{eq:main_ineq}
\hatS(\hatte) + \sum_{M \in \cM} \hatte_M U_M \leq \hatS (\te) + \sum_{M \in \cM} \te_M U_M
\ee
where $U_M =  4 \nu^{-2}_r\, \sigma^2 (\delta + 1)  \ \|\Psi_M\|^2_F\ \ln p$ is defined in Lemma~\ref{lem:lemma4}.
Then, combination of \eqref{eq:S_rel} and \eqref{eq:main_ineq} yields
\be \label{eq:for10} 
S(\hatte) -  S(\te) \leq    \sum_{M \in \cM} (\te_M - \hatte_M) U_M +
2 \sigma\, \langle \xi, \hat{f}_{\hatte}- \hat{f}_{\te} \rangle
\ee
Fix $\beta \in (0,1)$, $M_0 \in \Theta_{\cM}$ and let
 $e_M$ be the $M$-th vector of the canonical basis.
Consider
\be   \label{eq:hatf_te}
\te   = (1-\beta)\hatte + \beta e_{M_0}, \quad
\hatf_{\te} = (1-\beta) \hat{f}_{\hatte} + \beta\hatf_{M_0}.
\ee
and, due to
$\|f - \hatf_{M_0}\|^2 - \|f - \hat{f}_{\hatte}\|^2 - \|\hat{f}_{\hatte} - \hatf_{M_0}\|^2
= 2 \langle f - \hatf_{\hatte}, \hatf_{\hatte} - \hatf_{M_0}\rangle,
$
one can write
$$ 
\|(1-\beta) (f-\hat{f}_{\hatte}) + \beta (f - \hatf_{M_0}) \|^2
= (1-\beta) \|f-\hat{f}_{\hatte}\|^2 - \beta (1-\beta) \|\hatf_{\hatte} - \hatf_{M_0}\|^2
+ \beta \|f - \hatf_{M_0}\|^2
$$
Combining the last equality with \eqref{eq:hatf_te} obtain
$$  
\|f - \hatf_{\hatte}\|^2 - \|f - \hatf_{\te}\|^2 =
\beta  \lkv \|f - \hat{f}_{\hatte}\|^2 + \|\hat{f}_{\hatte} - \hatf_{M_0}\|^2 %
- \|f - \hatf_{M_0}\|^2\rkv - \beta^2  \|\hat{f}_{\hatte} - \hatf_{M_0}\|^2.
$$
Plugging the last equality into \eqref{eq:S_rel_main} derive
\be  \label{eq:S_difference}
\frac{1}{\beta}\, \lkv S(\hatte) - S(\te) \rkv =
\frac{1}{2}\,\|f - \hat{f}_{\hatte}\|^2 - \frac{1}{2}\,\|f - \hatf_{M_0}\|^2  +
\frac{1-\beta}{2}\, \|\hat{f}_{\hatte} - \hatf_{M_0}\|^2 +  \frac{1}{2}\, \Delta_1
\ee
where, due to  \eqref{eq:hatf_te},
\bes 
\Delta_1 = \frac{1}{\beta} \, \sum_{M \in \cM} (\hatte_M - \te_M) \|f - \hatf_M\|^2
= \sum_{M \in \cM} \hatte_M \|f - \hatf_M\|^2 - \|f -\hatf_{M_0}\|^2.
\ees
Now rewrite  $\Delta_1$ as
$$
\Delta_1 = \sum_{M \in \cM} \hatte_M  \|\hatf_{\hatte} - \hatf_M\|^2.
$$
using  identity \eqref{eq:lem3} of Lemma~\ref{lem:lemma3} with $\tilde{f} = f$.
Combining the last   formulae with \eqref{eq:S_difference}, we arrive at
\be \label{eq:for19} 
\frac{1}{\beta}\, \lkv S(\hatte) - S(\te) \rkv =
 \|f - \hat{f}_{\hatte}\|^2 - \|f - \hat{f}_{\hatte}\|^2 +
\frac{1}{2} \lfi \sum_{M \in \cM} \hatte_M \|\hf_{\hte} - \hf_M\|^2 + (1-\beta)  \|\hat{f}_{\hatte} - \hatf_{M_0}\|^2  \rfi.
\ee
Now, using $\hatf_{\te}$ given by \eqref{eq:hatf_te}, derive
\be \label{eq:for20} 
2 \sigma\, \langle \xi, \hat{f}_{\hatte}- \hat{f}_{\te} \rangle = 2 \beta \sigma\, \langle \xi, \hat{f}_{\hatte}- \hat{f}_{M_0} \rangle
\ee
Similarly,
\be \label{eq:for21} 
\sum_{M \in \cM} (\hatte_M - \te_M) U_M = \beta \lkv \sum_{M \in \cM}  \hatte_M U_M - U_{M_0} \rkv.
\ee
Plugging  \eqref{eq:for19}--\eqref{eq:for21} into \eqref{eq:for10} and setting $\beta \to 0$, arrive at
\be \label{eq:for22} 
\|f - \hat{f}_{\hatte}\|^2  \leq \|f - \hf_{M_0}\|^2 + U_{M_0} - \sum_{M \in \cM}  \hatte_M U_M +
2 \sigma \, \langle \xi, \hat{f}_{\hatte}- \hat{f}_{M_0} \rangle - \frac{1}{2}\, \Delta_2
\ee
where, applying Lemma~\ref{lem:lemma3} with $\tilde{f} = \hf_{M_0}$, we obtain
\bes
\Delta_2 =  \sum_{M \in \cM}  \hatte_M \|\hatf_{\hte} - \hatf_M\|^2 + \|\hatf_{\hte} - \hatf_{M_0}\|^2
=  \sum_{M \in \cM}  \hatte_M \|\hatf_{M} - \hatf_{M_0}\|^2
\ees
Recall  that
\bes
\|\hf_{M_0} - f\|^2 = \sigma^2 \xi^T H_{M_0} \xi + \|f_{M_0} - f\|^2 \leq
 \sigma^2\, \nu_r^{-2}\, \| \Psi_{M_0}^T\, \eps\|^2 + \|f_{M_0} - f\|^2
\ees
and re-write \eqref{eq:for22} as
\be \label{eq:for24} 
\|f - \hat{f}_{\hatte}\|^2 \leq \|f - f_{M_0}\|^2 + \frac{5}{2} U_{M_0} + \Delta_0
\ee
where $\Delta_0$ is defined by formula \eqref{eq:Delta0} of Lemma~\ref{lem:lemma4}.
Now, in order to complete the proof apply Lemma~\ref{lem:lemma4} and note that the set $M_0 \in \cM$ is arbitrary.
\newline
$\Box$

\end{document}